\documentclass[aps, pra, reprint, superscriptaddress]{revtex4-2}

\pdfoutput=1

\usepackage{header}

\graphicspath{{graphics/}}

\begin{document}

\bibliographystyle{apsrev4-2}

\title{Simulating Quantum Computations with Tutte Polynomials}

\author{Ryan L. Mann}
\email{mail@ryanmann.org}
\homepage{http://www.ryanmann.org}
\affiliation{School of Mathematics, University of Bristol, Bristol, BS8 1UG, United Kingdom}
\affiliation{Centre for Quantum Computation and Communication Technology, \\ Centre for Quantum Software and Information, \\ Faculty of Engineering \& Information Technology, University of Technology Sydney, NSW 2007, Australia}

\begin{abstract}
We establish a classical heuristic algorithm for exactly computing quantum probability amplitudes. Our algorithm is based on mapping output probability amplitudes of quantum circuits to evaluations of the Tutte polynomial of graphic matroids. The algorithm evaluates the Tutte polynomial recursively using the deletion-contraction property while attempting to exploit structural properties of the matroid. We consider several variations of our algorithm and present experimental results comparing their performance on two classes of random quantum circuits. Further, we obtain an explicit form for Clifford circuit amplitudes in terms of matroid invariants and an alternative efficient classical algorithm for computing the output probability amplitudes of Clifford circuits.
\end{abstract}

\maketitle

\section{Introduction}
\label{section:Introduction}

There is a natural relationship between quantum computation and evaluations of Tutte polynomials~\cite{aharonov2007polynomial, shepherd2010binary}. In particular, quantum probability amplitudes are proportional to evaluations of the Tutte polynomial of graphic matroids. In this paper we use this relationship to establish a classical heuristic algorithm for exactly computing quantum probability amplitudes. While this problem is known to be \mbox{\textsc{\#P}-hard} in general~\cite{fenner1999determining}, our algorithm focuses on exploiting structural properties of an instance to achieve an improved runtime over traditional methods. Previously it was known that this problem can be solved in time exponential in the treewidth of the underlying graph~\cite{markov2008simulating}.

The basis of our algorithm is a mapping between output probability amplitudes of quantum circuits and evaluations of the Tutte polynomial of graphic matroids~\cite{shepherd2009temporally, shepherd2010binary, bremner2010classical}. Our algorithm proceeds to evaluate the Tutte polynomial recursively using the deletion-contraction property. At each step in the recursion, our algorithm computes certain structural properties of the matroid in order to attempt to prune the computational tree. This approach to computing Tutte polynomials was first studied by Haggard, Pearce, and Royle~\cite{haggard2010computing}. Our algorithm can be seen as an adaption of their work to special points of the Tutte plane where we can exploit additional structural properties.

The performance of algorithms for computing Tutte polynomials based on the deletion-contraction property depends on the heuristic used to decide the ordering of the recursion~\cite{pearce2009edge, haggard2010computing, monagan2012new}. We consider several heuristics introduced by Pearce, Haggard, and Royle~\cite{pearce2009edge} and an additional heuristic, which is specific to our algorithm. We present some experimental results comparing the performance of these heuristics on two classes of random quantum circuits corresponding to dense and sparse instances.

The correspondence between output probability amplitudes of quantum circuits and evaluations of Tutte polynomials also allows us to obtain an explicit form for Clifford circuit amplitudes in terms of matroid invariants by a theorem of Pendavingh~\cite{pendavingh2014evaluation}. This gives rise to an alternative efficient classical algorithm for computing output probability amplitudes of Clifford circuits.

This paper is structured as follows. We introduce matroid theory in \mbox{Section~\ref{section:MatroidTheory}} and the Tutte polynomial in \mbox{Section~\ref{section:TuttePolynomial}}. In \mbox{Sections~\ref{section:PottsModelPartitionFunction}, \ref{section:InstantaneousQuantumPolynomialTime}, and \ref{section:QuantumComputationAndTheTuttePolynomial}}, we establish a mapping between output probability amplitudes of quantum circuits and evaluations of the Tutte polynomial of graphic matroids. This is achieved by introducing the Potts model partition function in \mbox{Section~\ref{section:PottsModelPartitionFunction}}, Instantaneous Quantum Polynomial-time circuits in \mbox{Section~\ref{section:InstantaneousQuantumPolynomialTime}}, and a class of universal quantum circuits in \mbox{Section~\ref{section:QuantumComputationAndTheTuttePolynomial}}. In \mbox{Section~\ref{section:EfficientClassicalSimulationCliffordCircuits}}, we use this mapping to obtain an explicit form for Clifford circuit amplitudes in terms of matroid invariants. We also obtain an efficient classical algorithm for computing the output probability amplitudes of Clifford circuits. We describe our algorithm in \mbox{Section~\ref{section:AlgorithmOverview}} and present some experimental results in \mbox{Section~\ref{section:ExperimentalResults}}. Finally, we conclude in \mbox{Section~\ref{section:ConclusionAndOutlook}}.

\section{Matroid Theory}
\label{section:MatroidTheory}

We shall now briefly introduce the theory of matroids. The interested reader is referred to the classic textbooks of Welsh~\cite{welsh1976matroid} and Oxley~\cite{oxley2006matroid} for a detailed treatment. Matroids were introduced by Whitney~\cite{whitney1935abstract} as a structure that generalises the notion of linear dependence. There are many equivalent ways to define a matroid. We shall define a matroid by the independence axioms.
\begin{definition}[Matroid]
    A matroid is a pair \mbox{$M=(\mathcal{S},\mathcal{I})$} consisting of a finite set $\mathcal{S}$, known as the \emph{ground set}, and a collection $\mathcal{I}$ of subsets of $\mathcal{S}$, known as the \emph{independent sets}, such that the following axioms are satisfied.
    \begin{enumerate}
        \item The empty set is a member of $\mathcal{I}$.
        \item Every subset of a member of $\mathcal{I}$ is a member of $\mathcal{I}$.
        \item If $A$ and $B$ are members of $\mathcal{I}$ and \mbox{$\abs{A}>\abs{B}$}, then there exists an \mbox{$x \in A {\setminus} B$} such that \mbox{$B\cup\{x\}$} is a member of $\mathcal{I}$.
    \end{enumerate}
\end{definition}

The rank of a subset $A$ of $\mathcal{S}$ is given by the \emph{rank function} \mbox{$r:2^{\mathcal{S}}\to\mathbb{N}$} of the matroid defined by \mbox{$r(A)\coloneqq\max\left(\abs{X} \mid X \subseteq A,X\in\mathcal{I}\right)$}. The \emph{rank} of a matroid $M$, denoted \mbox{$r(M)$}, is the rank of the set $S$.

The archetypal class of matroids are \emph{vector matroids}. A vector matroid \mbox{$M=(\mathcal{S},\mathcal{I})$} is a matroid whose ground set $\mathcal{S}$ is a subset of a vector space over a field $\mathbb{F}$ and whose independent sets $\mathcal{I}$ are the linearly independent subsets of $\mathcal{S}$. The rank of a subset of a vector matroid is the dimension of the subspace spanned by the corresponding vectors. We say that a matroid is \mbox{$\mathbb{F}$-\emph{representable}} if it is isomorphic to a vector matroid over the field $\mathbb{F}$. A matroid is a \emph{binary matroid} if it is \mbox{$\mathbb{F}_2$-representable} and is a \emph{ternary matroid} if it is \mbox{$\mathbb{F}_3$-representable}. A matroid that is representable over every field is called a \emph{regular matroid}.

Every finite graph \mbox{$G=(V,E)$} induces a matroid \mbox{$M(G)=(\mathcal{S},\mathcal{I})$} as follows. Let the ground set $\mathcal{S}$ be the set of edges $E$ and let the independent sets $\mathcal{I}$ be the subsets of $E$ that are a forest, i.e., they do not contain a simple cycle. It is easy to check that $M(G)$ satisfies the independence axioms. The rank of a subset $A$ of a cycle matroid is \mbox{$\abs{V}-\kappa(A)$}, where $\kappa(A)$ denotes the number of connected components of the subgraph with edge set $A$. The rank of the cycle matroid $M(G)$, denoted \mbox{$r(M(G))$} or simply $r(G)$, is the rank of the set $E$. The matroid $M(G)$ is called the \emph{cycle matroid} of $G$. We say that a matroid is \emph{graphic} if it is isomorphic to the cycle matroid of a graph.

Graphic matroids are regular. To see this consider assigning to the graph $G$ an arbitrary \emph{orientation} $D(G)$, that is, for each edge \mbox{$e=\{u,v\}$} in $G$, we choose one of $u$ and $v$ to be the \emph{positive end} and the other one to be the \emph{negative end}. Then construct the \emph{oriented incidence matrix} of $G$ with respect to the orientation $D(G)$.
\begin{definition}[Oriented incidence matrix]
    Let \mbox{$G=(V,E)$} be a graph and let $D(G)$ be an orientation of $G$. Then the oriented incidence matrix of $G$ with respect to $D(G)$ is the \mbox{$\abs{V}\times\abs{E}$} matrix \mbox{$A_{D(G)}=(a_{ve})_{\abs{V}\times\abs{E}}$} whose entries are
    \begin{equation}
        a_{ve} =
        \begin{dcases*}
            +1, & if $v$ is the positive end of $e$; \\
            -1, & if $v$ is the negative end of $e$; \\
            0, & otherwise.
        \end{dcases*} \notag
    \end{equation}
\end{definition}
The rows of the oriented incidence matrix $A_{D(G)}$ correspond to the vertices of $G$ and the columns correspond to the edges of $G$. Each column contains exactly one $+1$ and exactly one $-1$ representing the positive and negative ends of the corresponding edge. If the column space of $A_{D(G)}$ is the ground set of a vector matroid, then it is easy to see that a subset is independent if and only if it is a forest in $G$. Hence, the oriented incidence matrix provides a representation of a graphic matroid over every field.

A \emph{minor} of a matroid $M$ is a matroid that is obtained from $M$ by a sequence of \emph{deletion} and \emph{contraction} operations.
\begin{definition}[Deletion]
    Let \mbox{$M=(\mathcal{S},\mathcal{I})$} be a matroid and let $e$ be an element of the ground set. Then the deletion of $M$ with respect to $e$ is the matroid \mbox{$M{\setminus}\{e\}=(\mathcal{S}',\mathcal{I}')$} whose ground set is \mbox{$\mathcal{S}'=\mathcal{S}{\setminus}\{e\}$} and whose independent sets are \mbox{$\mathcal{I}'=\{I\subseteq\mathcal{S}{\setminus}\{e\} \mid I\in\mathcal{I}\}$}.
\end{definition}
The deletion of an element from the cycle matroid of a graph corresponds to removing an edge from the graph.
\begin{definition}[Contraction]
    Let \mbox{$M=(\mathcal{S},\mathcal{I})$} be a matroid and let $e$ be an element of the ground set. Then the contraction of $M$ with respect to $e$ is the matroid \mbox{$M/\{e\}=(\mathcal{S}',\mathcal{I}')$} whose ground set is \mbox{$\mathcal{S}'=\mathcal{S}{\setminus}\{e\}$} and whose independent sets are \mbox{$\mathcal{I}'=\{I\subseteq\mathcal{S}{\setminus}\{e\} \mid I\cup\{e\}\in\mathcal{I}\}$}.
\end{definition}
The contraction of an element from the cycle matroid of a graph corresponds to removing an edge from the graph and merging its two endpoints.

An element $e$ of a matroid is said to be a \emph{loop} if $\{e\}$ is not an independent set and said to be a \emph{coloop} if $e$ is contained in every maximally independent set. If an element $e$ of a matroid is either a loop or a coloop then the deletion and contraction of $e$ are equivalent.

\section{The Tutte Polynomial}
\label{section:TuttePolynomial}

We shall now briefly introduce the Tutte polynomial, which is a well-known invariant in matroid and graph theory.
\begin{definition}[Tutte polynomial of a matroid]
    Let \mbox{$M=(\mathcal{S},\mathcal{I})$} be a matroid with rank function \mbox{$r:2^{\mathcal{S}}\to\mathbb{N}$}. Then the Tutte polynomial of $M$ is the bivariate polynomial defined by
    \begin{equation}
        \mathrm{T}(M;x,y) \coloneqq \sum_{A \subseteq \mathcal{S}}(x-1)^{r(M)-r(A)}(y-1)^{\abs{A}-r(A)}. \notag
    \end{equation}
\end{definition}

The Tutte polynomial may also be defined recursively by the \emph{deletion-contraction property}.
\begin{definition}[Deletion-contraction property]
   Let \mbox{$M=(\mathcal{S},\mathcal{I})$} be a matroid. If $M$ is the empty matroid, i.e., \mbox{$\mathcal{S}=\varnothing$}, then
   \begin{equation}
        \mathrm{T}(M;x,y) = 1. \notag
    \end{equation}
    Otherwise, let $e$ be an element of the ground set. If $e$ is a loop, then
    \begin{equation}
        \mathrm{T}(M;x,y) = y\mathrm{T}(M{\setminus}\{e\};x,y). \notag
    \end{equation}
    If $e$ is a coloop, then
    \begin{equation}
        \mathrm{T}(M;x,y) = x\mathrm{T}(M/\{e\};x,y). \notag
    \end{equation}
    Finally, if $e$ is neither a loop nor a coloop, then
    \begin{equation}
        \mathrm{T}(M;x,y) = \mathrm{T}(M{\setminus}\{e\};x,y)+\mathrm{T}(M/\{e\};x,y). \notag
    \end{equation}
\end{definition}
The deletion-contraction property immediately gives an algorithm for recursively computing the Tutte polynomial. This algorithm is in general inefficient, but the performance may be improved by using isomorphism testing to reduce the number of recursive calls~\cite{sekine1995computing}. The performance of this algorithm depends on the heuristic used to choose elements of the ground set~\cite{pearce2009edge, haggard2010computing, monagan2012new}. Bj\"orklund et al.~\cite{bjorklund2008computing} showed that the Tutte polynomial can be computed in time exponential in the number of vertices.

The Tutte polynomial of a graph may be recovered by considering the Tutte polynomial of the cycle matroid of a graph and using the fact that the rank of a subset $A$ of a cycle matroid is \mbox{$\abs{V}-\kappa(A)$}, where $\kappa(A)$ denotes the number of connected components of the subgraph with edge set $A$.
\begin{definition}[Tutte Polynomial of a graph]
    Let \mbox{$G=(V,E)$} be a graph and let $\kappa(A)$ denote the number of connected components of the subgraph with edge set $A$. Then the Tutte polynomial of $G$ is a polynomial in $x$ and $y$, defined by
    \begin{equation}
        \mathrm{T}(G;x,y) \coloneqq \sum_{A \subseteq E}(x-1)^{\kappa(A)-\kappa(E)}(y-1)^{\kappa(A)+\abs{A}-\abs{V}}. \notag
    \end{equation}
\end{definition}

The Tutte polynomial is trivial to evaluate along the hyperbola \mbox{$(x-1)(y-1)=1$} for any matroid. In the case of graphic matroids, Jaeger, Vertigan, and Welsh~\cite{jaeger1990computational} showed that the Tutte polynomial is \mbox{\textsc{\#P}-hard} to evaluate, except along this hyperbola and when \mbox{$(x,y)$} equals one of nine special points.
\begin{theorem}[Jaeger, Vertigan, and Welsh~\cite{jaeger1990computational}]
    The problem of evaluating the Tutte polynomial of a graphic matroid at an algebraic point in the \mbox{$(x,y)$-plane} is \mbox{\emph{\#P-hard}} except when \mbox{$(x-1)(y-1)=1$} or when \mbox{$(x,y)$} equals one of \mbox{$(1,1)$}, \mbox{$(-1,-1)$}, \mbox{$(0,-1)$}, \mbox{$(-1,0)$}, \mbox{$(i,-i)$}, \mbox{$(-i,i)$}, \mbox{$(j,j^2)$}, \mbox{$(j,j^2)$}, or \mbox{$(j^2,j)$}, where \mbox{$j=\exp(2\pi i/3)$}. In each of these exceptional cases the evaluation can be done in polynomial time.
\end{theorem}
Vertigan~\cite{vertigan1998bicycle} extended this result to vector matroids.
\begin{theorem}[Vertigan~\cite{vertigan1998bicycle}]
    The problem of evaluating the Tutte polynomial of a vector matroid over a field $\mathbb{F}$ at an algebraic point in the \mbox{$(x,y)$-plane} is \mbox{\emph{\#P-hard}} except when \mbox{$(x-1)(y-1)=1$}, \mbox{$(x,y)$} equals \mbox{$(1,1)$}, or when
    \begin{enumerate}
        \item \mbox{$\abs{\mathbb{F}}=2$} and \mbox{$(x,y)$} equals one of \mbox{$(-1,-1)$}, \mbox{$(0,-1)$}, \mbox{$(-1,0)$}, \mbox{$(i,-i)$}, or \mbox{$(-i,i)$};
        \item \mbox{$\abs{\mathbb{F}}=3$} and \mbox{$(x,y)$} equals one of \mbox{$(j,j^2)$} or \mbox{$(j^2,j)$}, where \mbox{$j=\exp(2\pi i/3)$}; or
        \item \mbox{$\abs{\mathbb{F}}=4$} and \mbox{$(x,y)$} equals \mbox{$(-1,-1)$}.
    \end{enumerate}
    In each of these exceptional cases, except when \mbox{$(x,y)$} equals \mbox{$(1,1)$}, the evaluation can be done in polynomial time.
\end{theorem}
Snook~\cite{snook2012counting} showed that when \mbox{$(x,y)$} equals \mbox{$(1,1)$} and $\mathbb{F}$ is either a finite field of fixed characteristic or a fixed infinite field, then evaluating the Tutte polynomial is \mbox{\textsc{\#P}-hard}. It is an open problem to understand the complexity of evaluating the Tutte polynomial at \mbox{$(1,1)$} over any fixed field.

\section{The Potts Model Partition Function}
\label{section:PottsModelPartitionFunction}

The Potts model is a statistical physical model described by an integer \mbox{$q\in\mathbb{Z}^+$} and a graph \mbox{$G=(V,E)$}, with the vertices representing spins and the edges representing interactions between them. A set of edge weights \mbox{$\{\omega_e\}_{e \in E}$} characterise the interactions and a set of vertex weights \mbox{$\{\upsilon_v\}_{v \in V}$} characterise the external fields at each spin. A configuration of the model is an assignment $\sigma$ of each spin to one of $q$ possible states. The \emph{Potts model partition function} is defined as follows.
\begin{definition}[Potts model partition function]
    Let \mbox{$q\in\mathbb{Z}^+$} be an integer and let \mbox{$G=(V,E)$} be a graph with the weights \mbox{$\Omega=\{\omega_e\}_{e \in E}$} assigned to its edges and the weights \mbox{$\Upsilon=\{\upsilon_v\}_{v \in V}$} assigned to its vertices. Then the \mbox{$q$-state} Potts model partition function is defined by
    \begin{equation}
        \mathrm{Z}_\mathrm{Potts}(G;q,\Omega,\Upsilon) \coloneqq \sum_{\sigma\in\mathbb{Z}_q^V}w_G(\sigma), \notag
    \end{equation}
    where
    \begin{equation}
        w_G(\sigma) = \exp\left(\sum_{\{u,v\}\in E}\omega_{\{u,v\}}\delta(\sigma_u,\sigma_v)+\sum_{v\in V}\upsilon_v\delta(\sigma_v)\right). \notag
    \end{equation}
\end{definition}

The Potts model partition function with an external field is equivalent to the zero-field case on an augmented graph \mbox{$G'=(V',E')$}. To construct $G'$ from $G$, for each of the connected components \mbox{$\{C_i\}_{i=1}^{\kappa(E)}$} of $G$ add a new vertex $u_i$ and for every vertex \mbox{$v \in V(C_i)$} add an edge \mbox{$e_v=\{u_i,v\}$} with the weight $\upsilon_v$ assigned to it. Then we have the following proposition.
\begin{proposition}[{restate=[name=restatement]PottsModelPartitionFunctionExternalField}]
    \label{proposition:PottsModelPartitionFunctionExternalField}
    \begin{equation}
        \mathrm{Z}_\mathrm{Potts}(G;q,\Omega,\Upsilon) = q^{-\kappa(E)}\mathrm{Z}_\mathrm{Potts}(G';q,\Omega\cup\Upsilon,0). \notag
    \end{equation}
\end{proposition}
A similar proposition appears in Welsh's monograph~\cite{welsh1993complexity}; we prove \mbox{Proposition~\ref{proposition:PottsModelPartitionFunctionExternalField}} in \mbox{Appendix~\ref{section:PottsModelPartitionFunctionExternalField}}.

It will be convenient to consider the Potts model with weights that are all positive integer multiples of a complex number $\theta$. We shall implement this model on the augmented graph $G'$ with all weights equal to $\theta$ by replacing each edge with the appropriate number of parallel edges. Let us denote the partition function of this model by \mbox{$\mathrm{Z}_\mathrm{Potts}(G';q,\theta)$}. Then, we have the following proposition relating the partition function of this model to the Tutte polynomial of the augmented graph $G'$.
\begin{proposition}
    \label{proposition:PottsModelPartitionFunctionTuttePolynomialRelation}
    \begin{equation}
        \mathrm{Z}_\mathrm{Potts}(G';q,\theta) = q^{\kappa(E')}(e^\theta-1)^{r(G')}\mathrm{T}\left(G';x,y\right), \notag
    \end{equation}
    where \mbox{$x=\frac{e^\theta+q-1}{e^\theta-1}$} and \mbox{$y=e^\theta$}.
\end{proposition}
In particular, the \mbox{$q$-state} Potts model partition function is related to the Tutte polynomial along the hyperbola \mbox{$(x-1)(y-1)=q$}. For a proof of \mbox{Proposition~\ref{proposition:PottsModelPartitionFunctionTuttePolynomialRelation}}, we refer the reader to Welsh's monograph~\cite[Section 4.4]{welsh1993complexity}.

The \mbox{$2$-state} Potts model partition function specialises to the \emph{Ising model partition function}.
\begin{definition}[Ising model partition function]
    Let \mbox{$G=(V,E)$} be a graph with the weights \mbox{$\Omega=\{\omega_e\}_{e \in E}$} assigned to its edges and the weights \mbox{$\Upsilon=\{\upsilon_v\}_{v \in V}$} assigned to its vertices. Then the Ising model partition function is defined by
    \begin{equation}
        \mathrm{Z}_\mathrm{Ising}(G;\Omega,\Upsilon) \coloneqq \sum_{\sigma\in\{-1,+1\}^V}w_G(\sigma), \notag
    \end{equation}
    where
    \begin{equation}
        w_G(\sigma) = \exp\left(\sum_{\{u,v\}\in E}\omega_{\{u,v\}}\sigma_u\sigma_v+\sum_{v\in V}\upsilon_v\sigma_v\right). \notag
    \end{equation}
\end{definition}
\begin{proposition}
    \label{proposition:PottsIsingRelation}
    \begin{equation}
        \mathrm{Z}_\mathrm{Potts}(G;2,\Omega,\Upsilon) = w_G\mathrm{Z}_\mathrm{Ising}\left(G;\frac{\Omega}{2},\frac{\Upsilon}{2}\right), \notag
    \end{equation}
    where \mbox{$w_G=\exp\left(\frac{1}{2}\sum_{e\in E}\omega_e+\frac{1}{2}\sum_{v\in V}\upsilon_v\right)$}.
\end{proposition}
\begin{proof}
    The proof follows from some simple algebra.
\end{proof}

\section{Instantaneous Quantum Polynomial Time}
\label{section:InstantaneousQuantumPolynomialTime}

We shall now briefly introduce the class of commuting quantum circuits, known as \emph{Instantaneous Quantum Polynomial-time} (IQP) circuits~\cite{shepherd2009temporally}. These circuits exhibit many interesting mathematical properties. In particular, the output probability amplitudes of IQP circuits are proportional to evaluations of the Tutte polynomial of binary matroids~\cite{shepherd2010binary}. IQP circuits comprise only gates that are diagonal in the \mbox{Pauli-X} basis and are described by an \emph{X-program}.
\begin{definition}[X-program]
    An X-program is a pair \mbox{$(P,\theta)$}, where \mbox{$P=(p_{ij})_{m \times n}$} is a binary matrix and \mbox{$\theta\in[-\pi,\pi]$} is a real angle. The matrix $P$ is used to construct a Hamiltonian of $m$ commuting terms acting on $n$ qubits, where each term in the Hamiltonian is a product of Pauli-X operators,
    \begin{equation}
        \mathrm{H}_{(P,\theta)} \coloneqq -\theta\sum_{i=1}^m\bigotimes_{j=1}^nX_j^{p_{ij}}. \notag
    \end{equation}
    Thus, the columns of $P$ correspond to qubits and the rows of $P$ correspond to interactions in the Hamiltonian.
\end{definition}

An X-program induces a probability distribution $\mathcal{P}_{(P,\theta)}$ known as an \emph{IQP distribution}.
\begin{definition}[$\mathcal{P}_{(P,\theta)}$]
    For an X-program \mbox{$(P,\theta)$} with \mbox{$P=(p_{ij})_{m \times n}$}, we define $\mathcal{P}_{(P,\theta)}$ to be the probability distribution over binary strings \mbox{$x\in\{0,1\}^n$}, given by
    \begin{equation}
        \textbf{Pr}[x] \coloneqq \abs{\psi_{(P,\theta)}(x)}^2, \notag
    \end{equation}
    where
    \begin{equation}
        \psi_{(P,\theta)}(x) = \bra{x}\exp\left(-i\mathrm{H}_{(P,\theta)}\right)\ket{0^n}. \notag
    \end{equation}
\end{definition}

The principal probability amplitude \mbox{$\psi_{(P,\theta)}(0^n)$} of an IQP distribution is directly related to an evaluation of the Tutte polynomial of the binary matroid whose ground set is the row space of $P$.
\begin{proposition}
    \label{proposition:IQPTuttePolynomialRelation}
    Let \mbox{$(P,\theta)$} be an X-program with \mbox{$P=(p_{ij})_{m \times n}$}. Let \mbox{$M=(\mathcal{S},\mathcal{I})$} be the binary matroid whose ground set $\mathcal{S}$ is the row space of $P$, then,
    \begin{equation}
        \psi_{(P,\theta)}(0) = e^{i\theta(r(M)-m)}(i\sin(\theta))^{r(M)}\mathrm{T}\left(M;x,y\right), \notag
    \end{equation}
    where \mbox{$x=-i\cot(\theta)$} and \mbox{$y=e^{2i\theta}$}.
\end{proposition}
A similar result may be obtained for the other probability amplitudes. This can easily be seen when \mbox{$\theta=\frac{\pi}{2k}$} for \mbox{$k\in\mathbb{Z}^+$}, by firstly letting \mbox{$P\|^kx$} be the matrix obtained from $P$ by appending $k$ rows identical to $x$, and then observing that \mbox{$\psi_{(P,\theta)}(x)=-i\psi_{(P\|^kx,\theta)}(0^n)$}. For a proof of \mbox{Proposition~\ref{proposition:IQPTuttePolynomialRelation}} and a treatment of the general $\theta$ case, we refer the reader to Ref.~\cite[Section 3]{shepherd2010binary}.

We shall consider X-programs that are induced by a weighted graph.
\begin{definition}[Graph-induced X-program]
    For a graph \mbox{$G=(V,E)$} with the weights \mbox{$\left\{\omega_e\in[-\pi,\pi]\right\}_{e \in E}$} assigned to its edges and the weights \mbox{$\left\{\upsilon_v\in[-\pi,\pi]\right\}_{v \in V}$} assigned to its vertices, we define the X-program induced by $G$ to be an X-program $\mathcal{X}_G$ such that
    \begin{equation}
        \mathrm{H}_{\mathcal{X}_G} = -\sum_{\{u,v\} \in E}\omega_{\{u,v\}}X_uX_v-\sum_{v \in V}\upsilon_vX_v. \notag
    \end{equation}
\end{definition}
Any X-program can be efficiently represented by a graph-induced X-program~\cite{shepherd2009temporally}. The principal probability amplitude \mbox{$\psi_{\mathcal{X}_G}(0^n)$} of the IQP distribution generated by a graph-induced X-program is directly related to the Ising model partition function of the graph with imaginary weights.
\begin{proposition}[{restate=[name=restatement]IQPIsingModelPartitionFunctionRelation}]
    \label{proposition:IQPIsingModelPartitionFunctionRelation}
    Let \mbox{$G=(V,E)$} be a graph with the weights \mbox{$\Omega=\left\{\omega_e\in[-\pi,\pi]\right\}_{e \in E}$} assigned to its edges and the weights \mbox{$\Upsilon=\left\{\upsilon_v\in[-\pi,\pi]\right\}_{v \in V}$} assigned to its vertices, then,
    \begin{equation}
        \psi_{\mathcal{X}_G}\left(0^{\abs{V}}\right) = \frac{1}{2^{\abs{V}}}\mathrm{Z}_\mathrm{Ising}(G;i\Omega,i\Upsilon). \notag
    \end{equation}
\end{proposition}
\mbox{Proposition~\ref{proposition:IQPIsingModelPartitionFunctionRelation}} is well known~\cite{iblisdir2014low, fujii2017commuting}; we provide a proof in \mbox{Appendix~\ref{section:IQPIsingModelPartitionFunctionRelation}}. It will be convenient to consider graph-induced X-programs \mbox{$\mathcal{X}_{G(\theta)}$} with weights that are all positive integer multiples of a real angle $\theta$. As in \mbox{Section~\ref{section:TuttePolynomial}}, this model can be implemented on the augmented graph \mbox{$G'=(V',E')$} with all weights equal to $\theta$ by replacing each edge with the appropriate number of parallel edges. Let us denote the graph-induced X-program of this model by \mbox{$\mathcal{X}_{G'(\theta)}$}. Then, we have the following proposition.
\begin{proposition}[{restate=[name=restatement]IQPAugmentedGraphRelation}]
    \label{proposition:IQPAugmentedGraphRelation}
    \begin{equation}
        \psi_{\mathcal{X}_{G(\theta)}}\left(0^{\abs{V}}\right) = \psi_{\mathcal{X}_{G'(\theta)}}\left(0^{\abs{V'}}\right). \notag
    \end{equation}
\end{proposition}
We prove \mbox{Proposition~\ref{proposition:IQPAugmentedGraphRelation}} in \mbox{Appendix~\ref{section:IQPAugmentedGraphRelation}}. We also have the following proposition relating the principal probability amplitude to the Tutte polynomial of the augmented graph.
\begin{proposition}[{restate=[name=restatement]IQPAugmentedGraphTuttePolynomialRelation}]
    \label{proposition:IQPAugmentedGraphTuttePolynomialRelation}
    \begin{equation}
        \psi_{\mathcal{X}_{G'}(\theta)}\left(0^{\abs{V'}}\right) = e^{i\theta\left(r(G')-\abs{E'}\right)}(i\sin(\theta))^{r(G')}\mathrm{T}\left(G';x,y\right), \notag
    \end{equation}
    where \mbox{$x=-i\cot(\theta)$} and \mbox{$y=e^{2i\theta}$}. \notag
\end{proposition}
We prove \mbox{Proposition~\ref{proposition:IQPAugmentedGraphTuttePolynomialRelation}} in \mbox{Appendix~\ref{section:IQPAugmentedGraphTuttePolynomialRelation}}. Notice that if we let \mbox{$M=(\mathcal{S},\mathcal{I})$} be the binary matroid whose ground set $\mathcal{S}$ is the column space of the orientated incidence matrix $A_{D(G')}$ of $G'$ with an arbitrary orientation $D(G')$ assigned to it, then we can use \mbox{Proposition~\ref{proposition:IQPTuttePolynomialRelation}} to obtain \mbox{Proposition~\ref{proposition:IQPAugmentedGraphTuttePolynomialRelation}}.

\section{Quantum Computation and the Tutte Polynomial}
\label{section:QuantumComputationAndTheTuttePolynomial}

In this section we show that quantum probability amplitudes may be expressed in terms of the evaluation of a Tutte polynomial. We achieve this by showing that output probability amplitudes of a class of universal quantum circuits are proportional to the principal probability amplitude of some IQP circuit.

It will be convenient to define the following gate set.
\begin{definition}[$\mathcal{G}_\theta$]
    For a real angle \mbox{$\theta\in[-\pi,\pi]$}, we define $\mathcal{G}(\theta)$ to be the gate set
    \begin{equation}
        \mathcal{G}_\theta \coloneqq \{H, e^{i\theta X}, e^{i\theta XX}\}, \notag
    \end{equation}
    where $H$ denotes the Hadamard gate.
\end{definition}
It is easy to see that the gate set $\mathcal{G}_{\frac{\pi}{4}}$ generates the Clifford group and the gate set $\mathcal{G}_{\frac{\pi}{8}}$ is universal for quantum computation.

In the IQP model it is easy to implement the gates $e^{i\theta X}$ and $e^{i\theta XX}$. So in order to implement the entire gate set $\mathcal{G}_\theta$, it remains to show that we can implement the Hadamard gate. This can be achieved by the use of postselection when \mbox{$\theta=\frac{\pi}{4k}$} for \mbox{$k\in\mathbb{Z}^+$}~\cite{bremner2010classical}. To apply a Hadamard gate to the target state $\ket{\alpha}_t$ consider the following Hadamard gadget. Firstly introduce an ancilla qubit in the state $\ket{0}_a$ and apply the gate \mbox{$e^{\frac{i\pi}{4}(\mathbb{I}-X)_t(\mathbb{I}-X)_a}$} to \mbox{$\ket{\alpha}_t\ket{0}_a$}. Then measure qubit $t$ in the computational basis and postselect on an outcome of $0$. The output state of this gadget is then \mbox{$H\ket{\alpha}_a$}.

We shall consider quantum circuits that comprise gates from the set $\mathcal{G}_{\frac{\pi}{4k}}$ for an integer \mbox{$k\in\mathbb{Z}^+$}. Let $C_{k,n,m}$ denote such a circuit that acts on $n$ qubits and comprises $m$ Hadamard gates. Further let $\mathcal{X}_G(C_{k,n,m})$ denote the graph-induced X-program that implements the circuit $C_{k,n,m}$ by replacing each of the $m$ Hadamard gates with the Hadamard gadget. Then we have the following proposition.
\begin{proposition}
    \label{proposition:QuantumCircuitAmplitudeTuttePolynomialRelation}
    \begin{equation}
        \bra{0^n}C_{k,n,m}\ket{0^n} = \sqrt{2}^m\psi_{\mathcal{X}_G(C_{k,n,m})}\left(0^{n+m}\right). \notag
    \end{equation}
\end{proposition}
\begin{proof}
    The proof follows immediately from the application of the Hadamard gadgets.
\end{proof}

Any quantum amplitude may therefore be expressed as the evaluation of a Tutte polynomial by \mbox{Proposition~\ref{proposition:IQPAugmentedGraphRelation}}, \mbox{Proposition~\ref{proposition:IQPAugmentedGraphTuttePolynomialRelation}}, and \mbox{Proposition~\ref{proposition:QuantumCircuitAmplitudeTuttePolynomialRelation}}.

\section{Efficient Classical Simulation of Clifford Circuits}
\label{section:EfficientClassicalSimulationCliffordCircuits}

In this section we show how the correspondence between quantum computation and evaluations of the Tutte polynomial provides an explicit form for Clifford circuit amplitudes in terms of matroid invariants; namely, the \emph{bicycle dimension} and \emph{Brown's invariant}. This gives rise to an efficient classical algorithm for computing the output probability amplitudes of Clifford circuits. We note that it was first observed by Shepherd~\cite{shepherd2010binary} that to compute the probability amplitude of a Clifford circuit, it is sufficient to evaluate the Tutte polynomial of a binary matroid at the point $(x, y)$ equals $(-i, i)$, which can be efficiently computed by Vertigan's algorithm~\cite{vertigan1998bicycle}. We proceed with some definitions.

Let $V$ be a linear subspace of $\mathbb{F}_2^n$. The bicycle dimension and Brown's invariant are defined as follows.
\begin{definition}[Bicycle dimension]
    The bicycle dimension of $V$ is defined by
    \begin{equation}
        d(V) \coloneqq \dim(V \cap V^\perp). \notag
    \end{equation}
\end{definition}
\begin{definition}[Brown's invariant]
    If $\abs{\mathrm{supp}(x)} \equiv 0 \pmod{4}$ for all $x \in V \cap V^\perp$, then Brown's invariant $\sigma(V)$ is defined to be the smallest integer such that
    \begin{equation}
        \sum_{x \in V}i^{\abs{\mathrm{supp}(x)}} = \sqrt{2}^{d(V)+\dim(V)}e^{\frac{i\pi}{4}\sigma(V)}. \notag
    \end{equation}
\end{definition}

The following theorem of Pendavingh~\cite{pendavingh2014evaluation} provides an explicit form for the Tutte polynomial of a binary matroid at $(-i, i)$ in terms of the bicycle dimension and Brown's invariant.
\begin{theorem}[Pendavingh~\cite{pendavingh2014evaluation}]
    \label{theorem:TuttePolynomialExplicitForm}
    Let $V$ be a linear subspace of $\mathbb{F}_2^{\mathcal{S}}$ and let $M(V)$ be the corresponding binary matroid with ground set $\mathcal{S}$. If \mbox{$\abs{\mathrm{supp}(x)} \equiv 0 \pmod{4}$} for all \mbox{$x \in V \cap V^\perp$}, then,
    \begin{equation}
        \mathrm{T}(M(V);-i,i) = \sqrt{2}^{d(V)}e^{\frac{i\pi}{4}(2\abs{S}-3r(M)-\sigma(V))}. \notag
    \end{equation}
    Otherwise, $\mathrm{T}(M(V);-i,i)=0$. Further, $\mathrm{T}(M(V);-i,i)$ can be evaluated in polynomial time.
\end{theorem}
As an immediate consequence of \mbox{Theorem~\ref{theorem:TuttePolynomialExplicitForm}}, we obtain an explicit form for Clifford circuit amplitudes in terms the bicycle dimension and Brown's invariant of the corresponding matroid. Furthermore, we obtain an efficient classical algorithm for computing the output probability amplitudes of Clifford circuits. For similar results of this flavour see Refs.~\cite{guan2019stabilizer, gosset2020fast}.

\section{Algorithm Overview}
\label{section:AlgorithmOverview}

We shall now use the correspondence between quantum computation and evaluations of the Tutte polynomial to establish a heuristic algorithm for computing quantum probability amplitudes. To compute a probability amplitude, it is sufficient to compute the Tutte polynomial of a graphic matroid at \mbox{$x=-i\cot\left(\frac{\pi}{4k}\right)$} and \mbox{$y=e^{\frac{i\pi}{2k}}$} for an integer $k\geq2$~\cite{shepherd2009temporally, shepherd2010binary, bremner2010classical}. Our algorithm will use the deletion-contraction property to recursively compute the Tutte polynomial. At each step in the recursion, the algorithm will compute certain structural properties of the graph in order to attempt to prune the computational tree. Our algorithm can be seen an adaption of the work of Haggard, Pearce, and Royle~\cite{haggard2010computing} to special points of the Tutte plane.

We note that our approach differs from tensor network-based methods, which involve the contraction of a graph with tensors assigned to its vertices. These methods have been used to simulate quantum computations while exploiting structural properties of the graph~\cite{markov2008simulating, ogorman2019parameterization, gray2021hyper, huang2020classical, pan2021simulating}. However, our approach allows us to exploit an alternative class of structural properties. We proceed by describing the key aspects of our algorithm.

\subsection{Multigraph Deletion-Contraction Formula}
\label{section:MultigraphDeletionContractionFormula}

To improve the performance of our algorithm, we shall use the following deletion-contraction formula for multigraphs.
\begin{proposition}
    \label{proposition:MultigraphDeletionContractionFormula}
    Let \mbox{$G=(V,E)$} be a multigraph and let $e$ be a multiedge of $G$ with multiplicity $\abs{e}$. If $e$ is a loop, then
    \begin{equation}
        \mathrm{T}(G;x,y) = y^{\abs{e}}\mathrm{T}(G{\setminus}\{e\};x,y). \notag
    \end{equation}
    If $e$ is a coloop, then
    \begin{equation}
        \mathrm{T}(G;x,y) = \left(x+\sum_{i=1}^{\abs{e}-1}y^i\right)\mathrm{T}(G/\{e\};x,y). \notag
    \end{equation}
    Finally, if $e$ is neither a loop nor a coloop, then
    \begin{equation}
        \mathrm{T}(G;x,y) = \mathrm{T}(G{\setminus}\{e\};x,y)+\left(\sum_{i=0}^{\abs{e}-1}y^i\right)\mathrm{T}(G/\{e\};x,y). \notag
    \end{equation}
\end{proposition}
\mbox{Proposition~\ref{proposition:MultigraphDeletionContractionFormula}} can easily be proven from the deletion-contraction formula by induction; we omit the proof. 

If $U$ is the underlying graph of $G$, then the number of recursive calls may be bounded by \mbox{$O\left(2^{\abs{E(U)}}\right)$}. Alternatively, we may bound the number of recursive calls in terms of the number of vertices plus the number of edges \mbox{$s=\abs{V(U)}+\abs{E(U)}$} in the underlying graph. The number of recursive calls $R_s$ is then bounded by \mbox{$R_s \leq R_{s-1}+R_{s-2}$}, which is precisely the Fibonacci recurrence. Hence the number of recursive calls is bounded by \mbox{$O\left(\phi^{\abs{V(U)}+\abs{E(U)}}\right)$}, where \mbox{$\phi=\frac{1+\sqrt{5}}{2}$} is the golden ratio~\cite{wilf2002algorithms}. A careful analysis shows that the number of recursive steps is bounded by \mbox{$O\left(\tau(U)\cdot\abs{E(U)}\right)$}, where $\tau(U)$ denotes the number of spanning trees in $U$~\cite{sekine1995computing}.

At each step in the recursion, we use the multigraph deletion-contraction formula to remove all multiedges that correspond to either a loop or a coloop in the underlying graph. This process contributes a multiplicative factor to the proceeding evaluation. Notice that when $G$ is a graph whose underlying graph is a looped forest, then every edge in the underlying graph is either a loop or a coloop. Hence, we obtain the following formula for the Tutte polynomial of $G$.
\begin{corollary}
    Let \mbox{$G=(V,E)$} be a multigraph whose underlying graph $U$ is a looped forest. Further, for each edge $e$ in $U$, let $\abs{e}$ denote its multiplicity in $G$. Then,
    \begin{equation}
        \mathrm{T}(G;x,y) = \prod_{\substack{e \in E(U) \\ \mathrm{loop}}}y^{\abs{e}}\prod_{\substack{e \in E(U) \\ \mathrm{coloop}}}\left(x+\sum_{i=1}^{\abs{e}-1}y^i\right). \notag
    \end{equation}
\end{corollary}
\begin{proof}
    The proof follows immediately from Proposition~\ref{proposition:MultigraphDeletionContractionFormula}.
\end{proof}

\subsection{Graph Simplification}
\label{section:GraphSimplification}

There are a number of techniques that we can use to simplify the graph at each step in the recursion. Firstly, we may remove any isolated vertices, since they do not contribute to the evaluation. 

Secondly, when \mbox{$x=-i\cot\left(\frac{\pi}{4k}\right)$} and \mbox{$y=e^{\frac{i\pi}{2k}}$} for an integer \mbox{$k\in\mathbb{Z}^+$}, we may replace each multiedge with a multiedge of equal multiplicity modulo $4k$. To account for this, we multiply the proceeding evaluation by a efficiently computable factor. Specifically, we invoke the following proposition.
\begin{proposition}[{restate=[name=restatement]GraphEdgeSimplification}]
    \label{proposition:GraphEdgeSimplification}
    Fix \mbox{$k\in\mathbb{Z}^+$}. Let \mbox{$G=(V,E)$} be a multigraph and let \mbox{$G'=(V',E')$} be the graph formed from $G$ by taking the multiplicity of each multiedge in $G$ modulo $4k$. Then,
    \begin{equation}
        \mathrm{T}(G;x,y) = \left(ie^{\frac{i\pi}{4k}}\sin\left(\frac{\pi}{4k}\right)\right)^{\kappa(E)-\kappa(E')}\mathrm{T}(G';x,y), \notag
    \end{equation}
    where \mbox{$x=-i\cot\left(\frac{\pi}{4k}\right)$} and \mbox{$y=e^{\frac{i\pi}{2k}}$}. \notag
\end{proposition}
We prove \mbox{Proposition~\ref{proposition:GraphEdgeSimplification}} in \mbox{Appendix~\ref{section:GraphEdgeSimplification}}.

\subsection{Vertigan Graphs}
\label{section:VertiganGraphs}

The Tutte polynomial of a multigraph whose edge multiplicities are all integer multiples of an integer \mbox{$k\in\mathbb{Z}^+$} may be evaluated at the point \mbox{$x=-i\cot\left(\frac{\pi}{4k}\right)$} and \mbox{$y=e^{\frac{i\pi}{2k}}$} in polynomial time. This can be seen by the following proposition.
\begin{proposition}[{restate=[name=restatement]TuttePolynomialVertiganGraph}]
    \label{proposition:TuttePolynomialVertiganGraph}
    Fix \mbox{$k\in\mathbb{Z}^+$}. Let $G=(V,E)$ be a multigraph whose edge multiplicities are all integer multiples of $k$. Further let \mbox{$G'=(V',E')$} be the graph formed from $G$ by taking the multiplicity of each multiedge in $G$ divided by $k$. Then,
    \begin{equation}
        \mathrm{T}(G;x,y) = \left(\sqrt{2}e^{\frac{i\pi(1-k)}{4k}}\sin\left(\frac{\pi}{4k}\right)\right)^{-r(G)}\mathrm{T}(G';-i,i), \notag
    \end{equation}
    where \mbox{$x=-i\cot\left(\frac{\pi}{4k}\right)$} and \mbox{$y=e^{\frac{i\pi}{2k}}$}. \notag
\end{proposition}
We prove \mbox{Proposition~\ref{proposition:TuttePolynomialVertiganGraph}} in \mbox{Appendix~\ref{section:TuttePolynomialVertiganGraph}} and note that this is a special consequence of the \emph{$k$-thickening} approach of Jaeger, Vertigan, and Welsh~\cite{jaeger1990computational}. The Tutte polynomial may then be efficiently computed by Vertigan's algorithm~\cite{vertigan1998bicycle}; we call such a multigraph a \emph{Vertigan graph}. We may therefore prune the computational tree whenever the graph is a Vertigan graph with respect to $k$. Note that this corresponds to quantum circuits comprising gates from the Clifford group.

\subsection{Connected Components}
\label{section:ConnectedComponents}

The Tutte polynomial factorises over components.
\begin{proposition}
    \label{proposition:TuttePolynomialConnectedComponents}
    Let \mbox{$G=(V,E)$} be a graph with connected components $C=\{C_i\}_{i=1}^k$, then,
    \begin{equation}
        \mathrm{T}(G;x,y) = \prod_{i=1}^k\mathrm{T}(C_i;x,y). \notag
    \end{equation}
\end{proposition}
\mbox{Proposition~\ref{proposition:TuttePolynomialConnectedComponents}} can easily be proven from the deletion-contraction formula; we omit the proof. At each step in the deletion-contraction recursion, if the graph is disconnected, then we may use this property to prune the computational tree and hence improve performance.

\subsection{Biconnected Components}
\label{section:BiconnectedComponents}

An identical result holds for biconnected components.
\begin{proposition}[Tutte~\cite{tutte1954contribution}]
    \label{proposition:TuttePolynomialBiconnectedComponents}
    Let \mbox{$G=(V,E)$} be a graph with biconnected components $B=\{B_i\}_{i=1}^k$, then,
    \begin{equation}
        \mathrm{T}(G;x,y) = \prod_{i=1}^k\mathrm{T}(B_i;x,y). \notag
    \end{equation}
\end{proposition}
\mbox{Proposition~\ref{proposition:TuttePolynomialBiconnectedComponents}} can easily be proven from the deletion-contraction formula. For a proof, we refer the reader to Ref.~\cite[Section 3]{tutte1954contribution}. 
Similarly to the connected component case, we may use this property to prune the computational tree and improve performance. Note that the biconnected components of a graph may be listed in time linear in the number of edges via depth-first search~\cite{hopcroft1973algorithm}.

\subsection{Multi-Cycles}
\label{section:MultiCycles}

The Tutte polynomial of a multigraph whose underlying graph is a cycle may be computed in polynomial time by invoking the following proposition.
\begin{proposition}[Haggard, Pearce, and Royle~\cite{haggard2010computing}]
    \label{proposition:TuttePolynomialMultiCycle}
    Let \mbox{$G=(V,E)$} be a multigraph whose underlying graph $U$ is an $n$-cycle with edges indexed by the positive integers. Further, for each edge $e$ in $U$, let $\abs{e}$ denote its multiplicity in $G$. Then,
    \begin{align}
        \mathrm{T}(G;x,y) &= \sum_{k=1}^{n-2}\left(\prod_{j=k+1}^ny_x\left(\abs{e_j}\right)\prod_{j=1}^{k-1}y_1\left(\abs{e_j}\right)\right) \notag \\
        &\quad+y_x\left(\abs{e_n}+\abs{e_{n-1}}\right)\prod_{j=1}^{n-2}y_1\left(\abs{e_j}\right), \notag
    \end{align}
    where \mbox{$y_x(j) \coloneqq x+\sum_{i=1}^{j-1}y^i$}.
\end{proposition}
\mbox{Proposition~\ref{proposition:TuttePolynomialMultiCycle}} can easily be proven from the deletion-contraction formula. For a proof, we refer the reader to Ref.~\cite[Theorem 4]{haggard2010computing}. We may use this proposition to prune the computational tree whenever the underlying graph is a cycle.

\subsection{Planar Graphs}
\label{section:PlanarGraphs}

The Tutte polynomial of a planar graph along the hyperbola \mbox{$(x-1)(y-1)=2$} may be evaluated in polynomial time via the Fisher-Kasteleyn-Temperley (FKT) algorithm~\cite{temperley1961dimer, kasteleyn1963dimer, kasteleyn1967graph}. We may therefore use this algorithm to prune the computational tree whenever the underlying graph is planar. Note that we may test whether a graph is planar in time linear in the number of vertices~\cite{hopcroft1974efficient}.

\subsection{Edge-Selection Heuristics}
\label{section:EdgeSelectionHeuristics}

The performance of our algorithm depends on the heuristic used to select edges. We shall consider six edge-selection heuristics: \emph{vertex order}, \emph{minimum degree}, \emph{maximum degree}, \emph{minimum degree sum}, \emph{maximum degree sum}, and \emph{non-Vertigan}. These edge-selection heuristics were first studied by Pearce, Haggard, and Royle~\cite{pearce2009edge}, with the exception of non-Vertigan, which is specific to our algorithm.

\emph{Vertex order}: The vertices of the graph are assigned an ordering. A multiedge is selected from those incident to the lowest vertex in the ordering and whose other endpoint is also the lowest vertex of any incident in the ordering. For contractions, the vertex inherits the lowest of the positions in the ordering.

\emph{Minimum degree}: A multiedge is selected from those incident to a vertex with minimal degree in the underlying graph. 

\emph{Maximum degree}: A multiedge is selected from those incident to a vertex with maximal degree in the underlying graph.

\emph{Minimum degree sum}: A multiedge is selected from those whose sum of degrees of its endpoints is minimal in the underlying graph.

\emph{Maximum degree sum}: A multiedge is selected from those whose sum of degrees of its endpoints is maximal in the underlying graph.

\emph{Non-Vertigan}: A multiedge is selected from those whose multiplicity is not an integer multiple of $k$; we call such a multiedge \emph{non-Vertigan}. Using this edge-selection heuristic, the number of recursive calls may be bounded by \mbox{$O\left(2^{\nu(G)}\right)$}, where $\nu(G)$ denotes the number of non-Vertigan multiedges in $G$. This is due to the fact that both the deletion and contraction operation reduce the number of non-Vertigan multiedges by at least one. We note that this is similar to the \emph{Sum-over-Cliffords} approach studied in Refs.~\cite{bravyi2016trading, bravyi2016improved, bravyi2019simulation}.

\subsection{Other Methods}
\label{section:OtherMethods}

There are many other methods that may improve the performance of our algorithm, which we do not study. We shall proceed by discussing some of these.

\emph{Isomorphism testing}: During the computation the graphs encountered and the evaluation of their Tutte polynomial is stored. At each recursive step, we test whether the graph is isomorphic to one already encountered, and if so, we use the evaluation of the isomorphic graph instead. Haggard, Pearce, and Royle~\cite{haggard2010computing} showed that isomorphism testing can lead to an improvement in the performance of computing Tutte polynomials. Note that this may not be as effective when the input is a multigraph.

\emph{Almost planar}: At each step in the recursion, we may test whether the graph is close to being planar, and if so, select edges in such a way that the deletion and contraction operations give rise to a planar graph. For example, if the graph is \emph{apex}, that is, it can be made planar by the removal of a single vertex, then we may select a multiedge incident to such a vertex. Similarly, if the underlying graph is \emph{edge apex} or \emph{contraction apex}~\cite{lipton2017six}, then we may select a multiedge such that the deletion or the contraction operation gives rise to a planar graph.

\emph{$k$-connected components}: Similarly to the connected and biconnected component case, we may compute the Tutte polynomial in terms of its $k$-connected components~\cite{andrzejak1997splitting, bonin2004tutte}.

\section{Experimental Results}
\label{section:ExperimentalResults}

In this section we present some experimental results comparing the performance of the edge-selection heuristics described in \mbox{Section~\ref{section:EdgeSelectionHeuristics}} on two classes of random quantum circuits. Our experiments were performed using SageMath 9.0~\cite{sage2020sagemath}. The source code and experimental data are available at Ref.~\cite{mann2021data}.

The first class we consider corresponds to random instances of IQP circuits induced by dense graphs. Specifically, an instance is an IQP circuit induced by a complete graph with edge weights chosen uniformly at random from the set \mbox{$\{\frac{m\pi}{8} \mid m\in\mathbb{Z}/8\mathbb{Z}\}$}. This class of IQP circuits is precisely that studied in Ref.~\cite{bremner2016average}, where it is conjectured that approximating the corresponding amplitudes up to a multiplicative error is \textsc{\#P}-hard on average.

The second class we consider corresponds to random instances of IQP circuits induced by sparse graphs. Specifically, an instance is an IQP circuit induced by a random graph where each of the possible edges is included independently with probability $1/2$ and with edge weights chosen uniformly at random from the set \mbox{$\{\frac{m\pi}{8} \mid m\in\mathbb{Z}/8\mathbb{Z}\}$}.

We run our algorithm using each of the edge-selection heuristics to compute the principal probability amplitude of $64$ random instances of both the dense and sparse class on $12$ vertices. The performance of each edge-selection heuristic is measured by counting the number of leaves in the computational tree. Our experimental data is presented in \mbox{Appendix~\ref{section:TablesExperimentalData}}. We find that the non-Vertigan edge-selection heuristic performs particularly well for the dense class and the maximum degree sum edge-selection heuristic performs particularly well for the sparse class.

\section{Conclusion \& Outlook}
\label{section:ConclusionAndOutlook}

We established a classical heuristic algorithm for exactly computing quantum probability amplitudes. Our algorithm is based on mapping output probability amplitudes of quantum circuits to evaluations of the Tutte polynomial of graphic matroids. The algorithm evaluates the Tutte polynomial recursively using the deletion-contraction property while attempting to exploit structural properties of the matroid. We considered several edge-selection heuristics and presented experimental results comparing their performance on two classes of random quantum circuits. Further, we obtained an explicit form for Clifford circuit amplitudes in terms of matroid invariants and an alternative efficient classical algorithm for computing the output probability amplitudes of Clifford circuits.
 
\section*{Acknowledgements}

We thank Michael Bremner, Adrian Chapman, Iain Moffatt, Ashley Montanaro, Rudi Pendavingh, and Dan Shepherd for helpful discussions. This research was supported by the QuantERA ERA-NET Cofund in Quantum Technologies implemented within the European Union's Horizon 2020 Programme (QuantAlgo project), EPSRC grants EP/L021005/1, EP/R043957/1, and EP/T001062/1, and the ARC Centre of Excellence for Quantum Computation and Communication Technology (CQC2T), project number CE170100012. Data are available at the University of Bristol data repository, data.bris, at \url{https://doi.org/10.5523/bris.kbhgclva863q21tjkqpyr5uvq}.

\onecolumngrid

\appendix

\section{Proof of Proposition~\ref*{proposition:PottsModelPartitionFunctionExternalField}}
\label{section:PottsModelPartitionFunctionExternalField}

\PottsModelPartitionFunctionExternalField*

\begin{proof}
    By definition,
    \begin{align}
        q^{\kappa(E)}\mathrm{Z}_\mathrm{Potts}(G;q,\Omega,\Upsilon) &= q^{\kappa(E)}\sum_{\sigma\in\mathbb{Z}_q^{V}}\exp\left(\sum_{\{u,v\}\in E}\omega_{\{u,v\}}\delta(\sigma_u,\sigma_v)+\sum_{v\in V}\upsilon_v\delta(\sigma_v)\right) \notag \\
        &= q^{\kappa(E)}\prod_{i=1}^{\kappa(E)}\sum_{\sigma\in\mathbb{Z}_q^{V\left(C_i\right)}}\exp\left(\sum_{\{u,v\}\in E(C_i)}\omega_{\{u,v\}}\delta(\sigma_u,\sigma_v)+\sum_{v\in V(C_i)}\upsilon_v\delta(\sigma_v)\right). \notag
    \end{align}
    Now, by combining terms that are invariant under a $\mathbb{Z}_q$ symmetry, we have
    \begin{align}
        q^{\kappa(E)}\mathrm{Z}_\mathrm{Potts}(G;q,\Omega,\Upsilon) &= \prod_{i=1}^{\kappa(E)}\sum_{\sigma\in\mathbb{Z}_q^{V\left(C_i\right)}}\exp\left(\sum_{\{u,v\} \in E(C_i)}\omega_{\{u,v\}}\delta(\sigma_u,\sigma_v)\right)\sum_{\sigma'\in\mathbb{Z}_q}\exp\left(\sum_{v \in V(C_i)}\upsilon_v\delta(\sigma_v,\sigma')\right)
        \notag \\
        &= \sum_{\sigma\in\mathbb{Z}_q^{V'}}\exp\left(\sum_{\{u,v\} \in E}\omega_{\{u,v\}}\delta(\sigma_u,\sigma_v)+\sum_{\{u,v\} \in E' {\setminus} E}\upsilon_v\delta(\sigma_u,\sigma_v)\right) \notag \\
        &= \mathrm{Z}_\mathrm{Potts}(G';q,\Omega\cup\Upsilon,0). \notag
    \end{align}
    This completes the proof.
\end{proof}

\section{Proof of Proposition~\ref*{proposition:IQPIsingModelPartitionFunctionRelation}}
\label{section:IQPIsingModelPartitionFunctionRelation}

\IQPIsingModelPartitionFunctionRelation*

\begin{proof}
    By definition,
    \begin{align}
        \mbox{$\psi_{\mathcal{X}_G}(0^{\abs{V}})$} &= \bra{0^{\abs{V}}}\exp\left(i\sum_{\{u,v\} \in E}\omega_{\{u,v\}}X_uX_v+i\sum_{v \in V}\upsilon_vX_v\right)\ket{0^{\abs{V}}} \notag \\
        &= \bra{+^{\abs{V}}}\exp\left(i\sum_{\{u,v\} \in E}\omega_{\{u,v\}}Z_uZ_v+i\sum_{v \in V}\upsilon_vZ_v\right)\ket{+^{\abs{V}}} \notag \\
        &= \frac{1}{2^{\abs{V}}}\sum_{x,y\in\{0,1\}^V}\bra{y}\exp\left(i\sum_{\{u,v\} \in E}\omega_{\{u,v\}}Z_uZ_v+i\sum_{v \in V}\upsilon_vZ_v\right)\ket{x} \notag \\
        &= \frac{1}{2^{\abs{V}}}\sum_{x\in\{0,1\}^V}\exp\left(i\sum_{\{u,v\} \in E}\omega_{\{u,v\}}(-1)^{x_u \oplus x_v}+i\sum_{v \in V}\upsilon_v(-1)^{x_v}\right) \notag \\
        &= \frac{1}{2^{\abs{V}}}\sum_{z\in\{-1,+1\}^V}\exp\left(i\sum_{\{u,v\} \in E}\omega_{\{u,v\}}z_uz_v+i\sum_{v \in V}\upsilon_vz_v\right) \notag \\
        &= \frac{1}{2^{\abs{V}}}\mathrm{Z}_\mathrm{Ising}(G;i\Omega,i\Upsilon). \notag
    \end{align}
    This completes the proof.
\end{proof}

\section{Proof of Proposition~\ref*{proposition:IQPAugmentedGraphRelation}}
\label{section:IQPAugmentedGraphRelation}

\IQPAugmentedGraphRelation*

\begin{proof}
    \begin{align}
        \psi_{\mathcal{X}_{G}(\theta)}\left(0^{\abs{V}}\right) &= \frac{1}{2^{\abs{V}}}\mathrm{Z}_\mathrm{Ising}(G;i\theta,i\theta)\quad\text{(by \mbox{Proposition~\ref{proposition:IQPIsingModelPartitionFunctionRelation}})} \notag \\
        &= \frac{1}{2^{\abs{V}}}e^{-i\theta\left(\abs{E}+\abs{V}\right)}\mathrm{Z}_\mathrm{Potts}(G;2,2i\theta,2i\theta)\quad\text{(by \mbox{Proposition~\ref{proposition:PottsIsingRelation}})} \notag \\
        &= \frac{1}{2^{\abs{V}+\kappa(E)}}e^{-i\theta\left(\abs{E}+\abs{V}\right)}\mathrm{Z}_\mathrm{Potts}(G';2,2i\theta,0)\quad\text{(by \mbox{Proposition~\ref{proposition:PottsModelPartitionFunctionExternalField}})} \notag \\
        &= \frac{1}{2^{\abs{V}+\kappa(E)}}\mathrm{Z}_\mathrm{Ising}(G';i\theta,0)\quad\text{(by \mbox{Proposition~\ref{proposition:PottsIsingRelation}})} \notag \\
        &= \psi_{\mathcal{X}_{G'}(\theta)}\left(0^{\abs{V'}}\right)\quad\text{(by \mbox{Proposition~\ref{proposition:IQPIsingModelPartitionFunctionRelation}})}. \notag
    \end{align}
    This completes the proof.
\end{proof}

\section{Proof of Proposition~\ref*{proposition:IQPAugmentedGraphTuttePolynomialRelation}}
\label{section:IQPAugmentedGraphTuttePolynomialRelation}

\IQPAugmentedGraphTuttePolynomialRelation*

\begin{proof}
    \begin{align}
        \psi_{\mathcal{X}_{G'}(\theta)}\left(0^{\abs{V'}}\right) &= \frac{1}{2^{\abs{V'}}}\mathrm{Z}_\mathrm{Ising}(G';i\theta,0)\quad\text{(by \mbox{Proposition~\ref{proposition:IQPIsingModelPartitionFunctionRelation}})} \notag \\
        &= \frac{1}{2^{\abs{V'}}}e^{-i\theta\abs{E'}}\mathrm{Z}_\mathrm{Potts}(G';2,2i\theta,0)\quad\text{(by \mbox{Proposition~\ref{proposition:PottsIsingRelation}})} \notag \\
        &= \frac{1}{2^{r(G')}}e^{-i\theta\abs{E'}}(e^{2i\theta}-1)^{r(G')}\mathrm{T}\left(G';-i\cot(\theta),e^{2i\theta}\right)\quad\text{(by \mbox{Proposition~\ref{proposition:PottsModelPartitionFunctionTuttePolynomialRelation}})} \notag \\
        &= e^{i\theta\left(r(G')-\abs{E'}\right)}(i\sin(\theta))^{r(G')}\mathrm{T}\left(G';-i\cot(\theta),e^{2i\theta}\right). \notag
    \end{align}
    This completes the proof.
\end{proof}

\section{Proof of Proposition~\ref*{proposition:GraphEdgeSimplification}}
\label{section:GraphEdgeSimplification}

\GraphEdgeSimplification*

\begin{proof}
    \begin{align}
        \mathrm{T}(G;x,y) &= e^{\frac{i\pi}{4k}\left(\abs{E}-r(G)\right)}\left(i\sin\left(\frac{\pi}{4k}\right)\right)^{-r(G)}\psi_{\mathcal{X}_{G}\left(\frac{\pi}{4k}\right)}\left(0^{\abs{V}}\right)\quad\text{(by \mbox{Proposition~\ref{proposition:IQPAugmentedGraphTuttePolynomialRelation}})} \notag \\
        &= e^{\frac{i\pi}{4k}\left(\abs{E}-r(G)\right)}\left(i\sin\left(\frac{\pi}{4k}\right)\right)^{-r(G)}\psi_{\mathcal{X}_{G'}\left(\frac{\pi}{4k}\right)}\left(0^{\abs{V'}}\right) \notag \\
        &= e^{\frac{i\pi}{4k}\left(\abs{E}-\abs{E'}\right)}\left(ie^{\frac{i\pi}{4k}}\sin\left(\frac{\pi}{4k}\right)\right)^{r(G')-r(G)}\mathrm{T}(G';x,y)\quad\text{(by \mbox{Proposition~\ref{proposition:IQPAugmentedGraphTuttePolynomialRelation}})} \notag \\
        &= \left(ie^{\frac{i\pi}{4k}}\sin\left(\frac{\pi}{4k}\right)\right)^{\kappa(E)-\kappa(E')}\mathrm{T}(G';x,y). \notag
    \end{align}
    This completes the proof.
\end{proof}

\section{Proof of Proposition~\ref*{proposition:TuttePolynomialVertiganGraph}}
\label{section:TuttePolynomialVertiganGraph}

\TuttePolynomialVertiganGraph*

\begin{proof}
    \begin{align}
        \mathrm{T}(G;x,y) &= e^{\frac{i\pi}{4k}\left(\abs{E}-r(G)\right)}\left(i\sin\left(\frac{\pi}{4k}\right)\right)^{-r(G)}\psi_{\mathcal{X}_{G}\left(\frac{\pi}{4k}\right)}\left(0^{\abs{V}}\right)\quad\text{(by \mbox{Proposition~\ref{proposition:IQPAugmentedGraphTuttePolynomialRelation}})} \notag \\
        &= e^{\frac{i\pi}{4k}\left(\abs{E}-r(G)\right)}\left(i\sin\left(\frac{\pi}{4k}\right)\right)^{-r(G)}\psi_{\mathcal{X}_{G'}\left(\frac{\pi}{4}\right)}\left(0^{\abs{V'}}\right) \notag \\
        &= \left(\sqrt{2}e^{\frac{i\pi(1-k)}{4k}}\sin\left(\frac{\pi}{4k}\right)\right)^{-r(G)}\mathrm{T}(G';-i,i)\quad\text{(by \mbox{Proposition~\ref{proposition:IQPAugmentedGraphTuttePolynomialRelation}})}. \notag
    \end{align}
    This completes the proof.
\end{proof}

\section{Tables of Experimental Data}
\label{section:TablesExperimentalData}

We present our experimental data in \mbox{Table~\ref{table:TableExperimentalDataDense}} and \mbox{Table~\ref{table:TableExperimentalDataSparse}}. The rows of the tables represent edge-selection heuristics and the columns represent quantities relating to the number of leaves in the computational trees.

\begin{table}[ht]
\centering
\begin{tabular}{lccccccc}
    \hline
     & \textbf{Sum} & \textbf{Mean} & \textbf{Mean Deviation} & \textbf{\#Empty} & \textbf{\#Vertigan} & \textbf{\#Multicycle} & \textbf{\#Planar} \\
    \hline
    \textbf{Non-Vertigan} & $8888920$ & $138889$ & $36225$ & $912$ & $2124367$ & $50097$ & $6713544$ \\
    \textbf{Vertex Order} & $37173344$ & $580834$ & $150982$ & $186$ & $167486$ & $553618$ & $36452054$ \\
    \textbf{Minimum Degree} & $57014650$ & $890854$ & $162094$ & $0$ & $353780$ & $1238429$ & $55422441$ \\
    \textbf{Maximum Degree} & $28604576$ & $446947$ & $119407$ & $950$ & $86125$ & $412215$ & $28105286$ \\
    \textbf{Minimum Degree Sum} & $50716183$ & $792440$ & $170360$ & $0$ & $243290$ & $1010284$ & $49462609$ \\
    \textbf{Maximum Degree Sum} & $10993306$ & $171770$ & $39409$ & $6971$ & $8998$ & $78030$ & $10899307$ \\
    \hline
\end{tabular}
\caption{Performance of the edge-selection heuristics on $64$ random instances of the dense class on $12$ vertices.}
\label{table:TableExperimentalDataDense}
\end{table}

\begin{table}[ht]
\centering
\begin{tabular}{lccccccc}
    \hline
     & \textbf{Sum} & \textbf{Mean} & \textbf{Mean Deviation} & \textbf{\#Empty} & \textbf{\#Vertigan} & \textbf{\#Multicycle} & \textbf{\#Planar} \\
    \hline
    \textbf{Non-Vertigan} & $63958$ & $999$ & $973$ & $30$ & $7994$ & $467$ & $55467$ \\
    \textbf{Vertex Order} & $93642$ & $1463$ & $1518$ & $88$ & $74$ & $813$ & $92667$ \\
    \textbf{Minimum Degree} & $412557$ & $6446$ & $6316$ & $0$ & $2158$ & $7285$ & $403114$ \\
    \textbf{Maximum Degree} & $91218$ & $1425$ & $1476$ & $16$ & $87$ & $933$ & $90182$ \\
    \textbf{Minimum Degree Sum} & $291763$ & $4559$ & $4630$ & $0$ & $1138$ & $4169$ & $286456$ \\
    \textbf{Maximum Degree Sum} & $50375$ & $787$ & $808$ & $14$ & $25$ & $415$ & $49921$ \\
    \hline
\end{tabular}
\caption{Performance of the edge-selection heuristics on $64$ random instances of the sparse class on $12$ vertices.}
\label{table:TableExperimentalDataSparse}
\end{table}

\twocolumngrid

\bibliography{bibliography}

\end{document}